\documentclass[manuscript,nonacm]{acmart}
\usepackage[utf8]{inputenc}
\usepackage[T1]{fontenc}
\usepackage{graphicx}
\usepackage{xcolor}
\usepackage[normalem]{ulem}
\usepackage{textcomp}
\usepackage{capt-of}
\usepackage[algo2e,linesnumbered]{algorithm2e}
\usepackage{comment}
\usepackage{paralist}
\usepackage{multirow}

\newif\ifdraft
\newif\ifhideproofs
\drafttrue
\hideproofstrue

\newcommand{\svp}[1]{\ifdraft{} \paragraph{TODO: }\textcolor{blue}{#1 -- Sriram} \\\fi}

\providecommand{\Pr}{\text{Pr}}
\renewcommand{\Pr}{\text{Pr}}

\newcommand{\eps}{\varepsilon}
\newcommand{\ot}{\tilde{O}}
\newcommand{\congest}{\textsc{Congest}}

\newcommand{\poly}{\operatorname{\text{{\rm poly}}}}

\newcommand{\KTZero}{KT-$0$\xspace}
\newcommand{\KTOne}{KT-$1$\xspace}
\newcommand{\KTTwo}{KT-$2$\xspace}
\newcommand{\KTRho}{KT-$\rho$\xspace}

\def\polylog{\operatorname{polylog}}

\setcopyright{none}
\begin{document}

\title[]{Can We Break Symmetry with $o(m)$ Communication?}
\author{Shreyas Pai}
\email{shreyas-pai@uiowa.edu}
\affiliation{
  \institution{The University of Iowa}
  \city{Iowa City}
  \state{IA}
  \country{USA}
}
\author{Gopal Pandurangan}
\authornote{Gopal Pandurangan was supported, in part, by NSF grants IIS-1633720, CCF-1717075,  CCF-1540512, and  BSF grant 2016419.}
\email{gopalpandurangan@gmail.com}
\affiliation{
  \institution{University of Houston}
  \city{Houston}
  \state{TX}
  \country{USA}
}
\author{Sriram V.~Pemmaraju}
\authornote{Sriram V.~Pemmaraju was supported, in part, by NSF grant IIS-1955939.}
\email{sriram-pemmaraju@uiowa.edu}
\affiliation{
  \institution{The University of Iowa}
  \city{Iowa City}
  \state{IA}
  \country{USA}
}
\author{Peter Robinson}
\authornote{Peter Robinson was partially supported by a grant from the Research Grants Council (HKSAR) [Project No. CityU 11213620], as well as by a grant from the City University of Hong Kong [Project No. 7200639/CS].}

\email{peter.robinson@cityu.edu.hk}
\affiliation{
  \institution{City University of Hong Kong}
  \country{Hong Kong SAR, China}
}

\begin{abstract}
We study the  communication cost (or \emph{message complexity}) of fundamental distributed symmetry
breaking problems, namely,  coloring and MIS. While significant progress
has been made in understanding and improving the running time of such problems, much less is known about the message complexity of these problems. In fact, all known algorithms need at least $\Omega(m)$ communication for these problems, where $m$ is the number of edges in the graph.
We address the following question in this paper: \emph{can we solve problems such as coloring and MIS using sublinear, i.e., $o(m)$ communication, and if so under what conditions?}

In a classical result, Awerbuch, Goldreich, Peleg, and Vainish [JACM 1990] showed that fundamental \emph{global} problems such as broadcast and  spanning tree construction require at least $\Omega(m)$ messages in the KT-\(1\) \congest\ model (i.e., \congest\ model in which nodes have initial
knowledge of the neighbors' \texttt{ID}'s) when algorithms are
restricted to be comparison-based (i.e., algorithms in which node \texttt{ID}'s can only be compared).
Thirty five years after this result,  King, Kutten, and Thorup [PODC 2015] showed that one can solve the above problems using $\tilde{O}(n)$ messages ($n$ is the number of nodes in the graph) in $\tilde{O}(n)$ rounds in the KT-\(1\) \congest\ model if \emph{non-comparison-based} algorithms are permitted.
An important implication of this result is that one can use the synchronous nature of the KT-\(1\) \congest\ model, using silence to convey information, and solve \emph{any} graph problem
using non-comparison-based algorithms with $\tilde{O}(n)$ messages, but this takes an
\emph{exponential} number of rounds. In the asynchronous model, even this is not possible.

In contrast, much less is known about the message complexity
of \emph{local} symmetry breaking problems such as coloring and MIS.
Our paper fills this gap by presenting the following results.
\begin{description}
\item[Lower bounds:] In the KT-\(1\) CONGEST model, we show that any comparison-based algorithm, even a randomized Monte-Carlo algorithm with
constant success probability, requires $\Omega(n^2)$ messages in the worst case to solve either  $(\Delta+1)$-coloring or MIS, \emph{regardless} of the number of rounds. We also show that
$\Omega(n)$ is a lower bound on the number of messages for any $(\Delta+1)$-coloring or MIS algorithm, even non-comparison-based, and even with nodes having initial knowledge of up to
a constant radius.
\item[Upper bounds:]
In the KT-\(1\) CONGEST model, we present the following randomized non-comparison-based algorithms for coloring that, with high probability, use $o(m)$ messages and run in polynomially many rounds.
\begin{description}
\item[(a)] A $(\Delta+1)$-coloring algorithm that
uses $\tilde{O}(n^{1.5})$ messages, while running in $\tilde{O}(D+\sqrt{n})$ rounds, where
$D$ is the graph diameter.
Our result also implies an \emph{asynchronous} algorithm  for $(\Delta+1)$-coloring with the same message bound but running in $\tilde{O}(n)$ rounds.
\item[(b)] For any constant $\varepsilon > 0$, a $(1+\varepsilon)\Delta$-coloring algorithm  that
uses $\tilde{O}(n/\varepsilon^2)$ messages, while running in $\tilde{O}(n)$ rounds.
\end{description}
If we increase our input knowledge slightly to radius 2, i.e., in the KT-\(2\) CONGEST model,
we obtain:
\begin{description}
\item[(c)] A randomized comparison-based  MIS algorithm that uses $\tilde{O}(n^{1.5})$ messages.
while running in $\tilde{O}(\sqrt{n})$ rounds.
\end{description}
\end{description}

While our lower bound results can be viewed as counterparts to the classical result of
Awerbuch, Goldreich, Peleg, and Vainish [JACM 90], but for local problems, our algorithms are the first-known algorithms for coloring and MIS that take $o(m)$ messages and run in polynomially many rounds.
\end{abstract}

\maketitle
\thispagestyle{empty}

\section{Introduction}
\label{section:introduction}

There has been significant interest over the last decade in obtaining
communication-efficient algorithms for fundamental problems in distributed computing. In the \congest\ model, which is a  message-passing  model with small-sized messages (typically $O(\log n)$-sized, where $n$
is the number of nodes in the network), communication cost is usually measured by the number of messages. In the so-called
\emph{clean network} model, a.k.a. the \KTZero (\textbf{K}nowledge \textbf{T}ill radius 0) model, where nodes have intial knowledge of only themselves and don't even know the \texttt{ID}'s of neighbors, Kutten et al.~\cite{kuttenjacm15} showed
that $\Omega(m)$ ($m$ is the number of edges in the network) is a lower bound for the message complexity for fundamental \emph{global} problems such as leader election, broadcast,  spanning tree, and mimimum spanning tree (MST) construction. This lower bound applies even for randomized Monte Carlo algorithms. For all these problems, there exist algorithms that (essentially) match this message lower bound; in fact, these also have optimal time complexity (of $D$, the network
diameter) in the \congest\ model  (see e.g., \cite{kuttenjacm15,stoc17,elkin}).

The clean network model does not capture many real world networks such as the Internet and peer-to-peer networks where nodes typically have knowledge of identities (i.e., IP addresses)
of other nodes. Also, there has been a lot of recent interest in ``all-to-all'' communication models such
as the congested clique \cite{LPPPSICOMP2005}, Massively Parallel Computing (MPC) \cite{karloff2010model}, and $k$-machine model \cite{KNPRSODA2015}, where each machine is assumed to have knowledge of \texttt{ID}'s of all other machines.
Motivated by these applications and models,
there has been a lot of recent interest
in studying message-efficient algorithms under the so-called \KTOne\ version of the \congest\ model,
where nodes have initial knowledge of the \texttt{ID}s of their neighbors, but no
other knowledge of their neighbors. An immediate question that arises is whether the $\Omega(m)$ message lower bound also holds
in the \KTOne model; or whether sublinear, i.e., $o(m)$ message complexity
is possible.

The above question was partially answered in a seminal paper by Awerbuch et al.~\cite{AwerbuchGPV1988} who initiated the study of trade-offs between the message complexity and \emph{initial knowledge} of distributed algorithms that solve \emph{global} problems, such as broadcast and spanning tree construction.
For any integer $\rho > 0$, in the \KTRho\
version of the \congest\ model (in short, \KTRho\ \congest), each node $v$ is provided initial  knowledge of (i) the \texttt{ID}s of all nodes at distance at most $\rho$ from $v$ and (ii) the neighborhood of every node at distance at most $\rho-1$ from $v$.
The bounds in this paper \cite{AwerbuchGPV1988} are for \emph{comparison-based algorithms}, i.e., algorithms in which
local computations on \texttt{ID}s are restricted to comparisons only.   This means that operations on \texttt{ID}s such as those used in the Cole-Vishkin coloring algorithm \cite{ColeVishkinSTOC1986} or applying random hash functions to \texttt{ID}s are disallowed.
Comparison-based algorithms are quite natural and indeed, most distributed algorithms (with few notable exceptions such as Cole-Vishkin \cite{ColeVishkinSTOC1986} and hash-functions based algorithms of King et al \cite{KingKTPODC2015}) are comparison-based.
For the \KTOne \congest\ model the authors show that $\Omega(m)$ messages are needed  for any \emph{comparison-based} algorithm (even randomized) that solves broadcast. Furthermore, in the KT-$\rho$ \congest\ model, $\Omega\left(\min\left\{m, n^{\frac{1+\Theta(1)}{\rho}}\right\}\right)$ messages are needed for any comparison-based algorithm that solves broadcast. The paper also shows matching upper bounds for comparison-based algorithms for broadcast.
These lower bounds also hold for non-comparison based algorithms,
where the size of the \texttt{ID}s is very large and grows independently with respect to message size, time, and randomness. This paper left open the possibility of circumventing the lower bound if one uses non-comparison based algorithms on more natural \texttt{ID} spaces typically used in distributed algorithms (as assumed in the current paper), where
\texttt{ID}s are drawn from a polynomial-sized \texttt{ID} space.

Nearly 35 years later, the above question was settled by King et al.~\cite{KingKTPODC2015} who showed that the Awerbuch et al.~lower bounds ``break'' if the assumption that the algorithms be comparison-based is dropped and one uses
\texttt{ID} space that is of polynomial size.\footnote{This can be relaxed to allow even exponential-sized \texttt{ID} space: by using fingerprinting technique \cite{karp-rabin,KingKTPODC2015}, with high probability, one can map $n$ \texttt{ID}s in exponential ID space to distinct \texttt{ID}s in polynomial
ID space.} Specifically, it is shown in \cite{KingKTPODC2015} that the Spanning Tree (and hence broadcast) and Minimum Spanning Tree (MST) problem can be solved using $\ot(n)$ messages in \KTOne \congest\ model.\footnote{We use $\ot(f(n))$ as short for $O(f(n) \cdot \poly \log n)$ and $\tilde{\Omega}(g(n))$ as short for $\Omega(g(n)/(\poly \log n))$.} In followup papers, it is shown that these problems can be solved with $o(m)$ messages, but with a higher message bound of $\ot(n^{1.5})$, even in the \emph{asynchronous}
\congest\ \KTOne model \cite{KMDISC18,KMDISC19}.
Using the King et al.~\cite{KingKTPODC2015} result,  it is
possible to solve \emph{any} graph problem (including symmetry breaking problems) using randomized \emph{non-comparison} based algorithms in $\tilde{O}(n)$ messages. However, this takes an
\emph{exponential} number of rounds. This is done by building
a spanning tree using the algorithm of King et al.~and then using time-encoding
to convey the entire topology to the root of the spanning tree. The root then locally computes the result and disseminates it to the entire network, again
using time-encoding (e.g., see \cite{petersoda} for details). Time-encoding uses silence to convey information and takes at least exponential (in $m$) rounds. Note that this works only in \emph{synchronous} setting and not in the asynchronous model. Hence, designing algorithms that use $\ot(n)$ (or even $o(m)$) messages for other graph problems, including local symmetry breaking problems, regardless of the number of rounds, in the asynchronous \congest\ \KTOne model is open.

Motivated by the above results, we initiate a similar study, but for fundamental \emph{local symmetry breaking} problems, such as $(\Delta+1)$-coloring and Maximal Independent Set (MIS).
These problems have been studied extensively for over four decades.
Significant progress has been made in understanding and improving the \emph{running time} (round complexity) of these problems  (see e.g., \cite{coloring-book,pettie,magnus,rozhon,ghaffarimis,GhaffariSODA2019,BEPSFOCS2012} and the references therein); however, much
less is known  with respect to message complexity.
For $(\Delta+1)$-coloring and MIS, to the best of our knowledge, all known distributed algorithms use at least $\Omega(m)$ messages. The overarching question we address in this paper is whether these problems can be solved using $o(m)$ messages in the \congest\ model and if so, under what conditions.

Our paper presents both negative and positive answers for the above question and shows results in three general directions.
First, we show that even though the \emph{round complexity} of local symmetry breaking problems is provably much smaller than the round complexity of global problems, comparison-based algorithms for local symmetry breaking problems require \emph{as many messages} as they do for global problems in the \KTOne\ \congest\ model.
Second, we show that if we drop the requirement that our algorithms be comparison-based only, then it is possible to design algorithms for local symmetry breaking problems in the \KTOne\ \congest\ model
that use far fewer messages.
Third, as we increase $\rho$, the radius of initial knowledge, to just two,
i.e., in the \KTTwo \congest\ model, it is possible to
design algorithms for local symmetry breaking problems that use even fewer messages.
The specific results that illustrate these three directions are presented in the next subsection.

\subsection{Main Results}
\label{subsection:mainResults}

\begin{figure}
\begin{center}
\def\arraystretch{1.5}
    \begin{tabular}{|l|l|l|}
        \hline
                                  & \textbf{\((\Delta + 1)\)-coloring}                 & \textbf{MIS}                                     \\
        \hline
        \multirow{2}{*}{KT-\(1\)} & Lower Bound (C): \(\Omega(m)\)            & Lower Bound (C): \(\Omega(m)\)          \\\cline{2-3}
                                  & Upper Bound (NC): \(\ot(n^{1.5})\)  & Upper Bound (C): \(\ot(m)\)           \\
        \hline
        \multirow{2}{*}{KT-\(2\)} & Lower Bound (NC): \(\Omega(n)\)   & Lower Bound (NC): \(\Omega(n)\)  \\\cline{2-3}
                                  &                & Upper Bound (C): \(\ot(n^{1.5})\)     \\
        \hline
        KT-\(\rho\)               & Lower Bound (NC): \(\Omega(n)\)   & Lower Bound (NC): \(\Omega(n)\) \\
        \hline
    \end{tabular}

\end{center}

\caption{A summary of lower and upper bounds results in this paper. The notation ``(C)'' and ``(NC)'' in each cell stand for comparison-based and non-comparison-based respectively. The \KTOne\ upper bound of $\ot(m)$ for MIS is not from this paper; it is immediately implied by a number of well-known MIS algorithms (e.g., \cite{LubySTOC1985,ghaffarimis}).
The \KTRho\ lower bounds hold for any constant $\rho \ge 1$
and hold even for non-comparison-based algorithms.
}
\label{fig:resultsTable}
\end{figure}

We present new lower and upper bounds on the message complexity
for two fundamental symmetry breaking problems, namely, coloring and MIS.
See Figure \ref{fig:resultsTable} for a summary.
\begin{description}
\item[Lower bounds:] In the \KTOne \congest\ model, we show that any comparison-based algorithm, even a randomized Monte Carlo algorithm with
constant success probability, requires $\Omega(n^2)$ messages in the worst case to solve either  $(\Delta+1)$-coloring or MIS, regardless of the number of rounds. Our result can be considered as a counterpart  to the classical result of Awerbuch et al.~\cite{AwerbuchGPV1988},
but for local problems.
We also show that in the \KTRho\ \congest\ model, for
any constant $\rho \ge 1$, $(\Delta+1)$-coloring and MIS require $\Omega(n)$ messages
even for non-comparison-based and Monte Carlo randomized algorithms with constant success probability.

\item[Upper bounds:] In the \KTOne \congest\ model, we present the following randomized non-comparison-based algorithms for coloring that
with high probability\footnote{This refers to probability at least $1 - n^{-c}$ for constant $c \ge 1$.} (w.h.p.) use $o(m)$ messages and run in polynomially many rounds.
\begin{description}
\item[(a)] A $(\Delta+1)$-coloring algorithm that
uses $\ot(n^{1.5})$ messages, while running in $\ot(D+\sqrt{n})$ rounds, where
$D$ is the graph diameter.
Our result also implies an \emph{asynchronous} algorithm  for $(\Delta+1)$-coloring with the same message bound but running in $\ot(n)$ rounds.
\item[(b)] A $(1+\eps)\Delta$-coloring algorithm  that
uses $\ot(n/\eps^2)$ messages, while running in $\ot(n)$ rounds.
\end{description}
If we increase our input knowledge slightly, i.e., we work in the \KTTwo \congest\ model, where nodes
have initial knowledge of their \emph{two hop-neighborhood}, then we
get the following additional and stronger result.
\begin{description}
\item[(c)] A \emph{comparison-based} algorithm for MIS that uses  $\ot(n^{1.5})$ messages, while running in $\ot(\sqrt{n})$ rounds.
\end{description}
Our algorithms for coloring and MIS  are the first-known algorithms that take
$o(m)$ messages and running in polynomial number of rounds.
\end{description}

\subsection{Other Related Work}
Several recent papers (see e.g., \cite{GmyrPanduranganDISC18, ghaffarikuhndisc18,KMDISC18,KMDISC19} have studied message-efficient algorithms for \emph{global} problems, namely, construction of spanning tree, minimum spanning tree, broadcasting and leader election, in the \KTOne\ \congest\ model inspired
by the work of King et al.~\cite{KingKTPODC2015}. We note that all these
are non-comparison-based algorithms. We use these prior
algorithms for our non-comparison-based algorithms in the \KTOne and \KTTwo models.
In a recent paper, Robinson \cite{petersoda} shows non-trivial lower bounds on the message complexity
of constructing graph spanners in the \congest\ \KTOne model.

In contrast to global problems, much less is known about obtaining sublinear, i.e.,
$o(m)$ algorithms for local problems, such as MIS and coloring.
Pai et al.~\cite{paidisc17} showed that MIS has a fundamental lower bound
of $\Omega(n^2)$ messages in the \congest\ \KTZero model (even for randomized algorithms). However, this result does not extend to the \KTOne model.
In contrast, they also showed that the 2-ruling set problem (note that MIS is 1-ruling set) can be solved using $\ot(n)$ messages in the \KTZero\ model in polynomial time.
To the best of our knowledge, we are not aware of other results on
the message complexity (in particular, lower bounds and sublinear upper bounds)
on fundamental symmetry breaking problems, vis-a-vis the initial input knowledge.

\subsection{Technical Contributions}
\label{subsection:technicalContributions}
\begin{itemize}
\item\textbf{Lower bounds:}
To obtain our \KTOne\ \congest\ lower bounds for comparison-based algorithms for $(\Delta+1)$-coloring and MIS, we start with the machinery introduced by Awerbuch et al.~\cite{AwerbuchGPV1988} for proving their \KTOne\ \congest\ lower bounds for comparison-based algorithms for broadcast.
At the core of their approach is an indistinguishability argument that uses edge crossings. Edge crossings have been used numerous times to prove
a variety of distributed computing lower bounds (see \cite{KorachMoranZaksSICOMP1987,kuttenjacm15, PaiPemmarajuFSTTCS2020,paidisc17,Patt-ShamirPerrySSS2017} for some examples). However, in the \KTOne\ \congest\ model, indistinguishability arguments via edge crossing are more challenging because when an edge incident on a node is crossed, the node is exposed to a new \texttt{ID} due to \KTOne.
For symmetry breaking problems, there is a further challenge due to the fact that multiple outputs are possible and the indistinguishability argument needs to work for all outputs.
Finally, since we want to show our lower bounds even for Monte Carlo algorithms with constant success probability, we require our indistinguishability arguments to apply to a large fraction of edge crossings (so as to be able to apply Yao's lemma \cite{YaoFOCS1977,MotwaniRaghavan}).
The lower bound graph family and \texttt{ID} assignment we design, overcomes all of these challenges. We use a unified construction that works for both $(\Delta+1)$-coloring and MIS and we expect this construction to work for other symmetry breaking problems such as maximal matching and edge coloring.

\item\textbf{Upper bounds:}
Our upper bounds are largely obtained by exploiting the fact that shared (or public) randomness combined with \KTOne\ is a powerful way of eliminating
the need to communicate over a large number of edges.\footnote{Note that we do not a priori assume shared randomness, but only private randomness (as is usual), but use the danner structure (Section \ref{sec:danner}) to share privately generated random bits throughput the graph.}
Specifically, we start with the recent coloring algorithm of Chang et al.~ \cite{ChangFGUZPODC2019} that works efficiently in the MPC model.
Roughly speaking, this algorithm starts with a probabilistic step; by randomly partitioning the nodes and the color palette. Then, after this probabilistic step, a large number of edges become inactive for the rest of algorithm.
This property is crucial to ensuring that the algorithm is efficient in the MPC model.
After the probabilistic step, nodes exchange their state with neighbors in so that every node can determine which of its incident edges to render inactive.
This state exchange is cheap in the MPC model, but is costly with respect to messages in the \congest\ model. We show how to simulate this coloring algorithm in the \congest\ model without the costly exchange of state. Instead we use shared randomness with limited dependence combined with \KTOne.
\end{itemize}

\subsection{Preliminaries}
\label{subsection:preliminaries}

\subsubsection{KT-$\rho$ \congest\ model}
\label{sec:model}
We work in the synchronous, message-passing model of distributed computing, known as the \congest\ model.
The input is a graph \(G=(V, E)\), $n = |V|$, which also serves as the communication network.
Nodes in the graph are processors with unique \texttt{ID}s from a space whose size is polynomial in $n$. In each round, each node can send an \(O(\log n)\)-bit message to each of its neighbors. Since we are interested in message complexity, the initial knowledge of the nodes is important.
For any integer $\rho > 0$, in the \KTRho\  \congest\ model
each node $v$ is provided initial  knowledge of (i) the \texttt{ID}s of all nodes at distance at most $\rho$ from $v$ and (ii) the neighborhood of every vertex at distance at most $\rho-1$ from $v$.

\subsubsection{Comparison-based Algorithms}
\label{sec:comp-model}
Often, the outcome of a distributed algorithm does not depend on specific values of
node \texttt{ID}s, but may depend on the relative ordering of \texttt{ID}s. For example, node \texttt{ID}s of endpoints may be used to break ties between edges of the same weight vying to join a minimum spanning tree. In this case, only the ordering of the \texttt{ID}s matters, not their
specific values.
Since this type of behavior is characteristic of many distributed algorithms,
Awerbuch et al.~\cite{AwerbuchGPV1988} formally define these as \emph{comparison-based} algorithms.
In comparison-based algorithms, the algorithm executed by each node contains two types of variables:
\texttt{ID}-type variables and \emph{ordinary} variables.
In the \KTRho\ \congest\ model, the \texttt{ID}-type variables at a node $v$ will store the \texttt{ID}s of all nodes within $\rho$ hops of $v$.
Nodes can send \texttt{ID}-type variables in messages, but since messages in the \congest\ model are restricted to be \(O(\log n)\) bits long, each message can contain only a constant number of \texttt{ID}-type variables. The local computations at any node may involve operations of the following two forms only:
\begin{enumerate}
  \item Comparing two \texttt{ID}-type variables \(I_{i}, I_{j}\) and storing the result of the comparison in an ordinary variable.
  \item Performing an arbitrary computation on ordinary variables and storing the result in another ordinary variable.
\end{enumerate}

Note that if randomization is allowed, then nodes can choose to ignore the node IDs and choose a new set of ($O(\log n)$-length) IDs and do arbitrary computations with them. These are still comparison-based algorithms.\footnote{However, note that such randomly chosen node IDs are unknown to neighbors and if the algorithm uses only those IDs
then this becomes effectively the KT0 model where bounds are already known \cite{paidisc17, AwerbuchGPV1988}.}

\subsubsection{Efficient Broadcasting in the \KTOne\ \congest\ model} \label{sec:danner}

As explained earlier, shared randomness along with initial knowledge, plays a key role in making our algorithms message-efficient.
We use a graph structure called a \emph{danner} introduced by Gmyr and Pandurangan \cite{GmyrPanduranganDISC18} to share random bits among the nodes in the graph in a message-efficient fashion.
Their specific result is stated in the following theorem.

\begin{theorem}[Gmyr and Pandurangan \cite{GmyrPanduranganDISC18}]
\label{lem:danner}
Given an $n$-vertex, $m$-edge, diameter $D$, graph $G = (V, E)$ and a parameter $\delta \in [0, 1]$, there is a randomized algorithm in the  KT-1 \congest\ model, that constructs a spanning subgraph (i.e., a danner) $H$ of $G$ such that $H$ has $\ot(\min\{m, n^{1+\delta}\})$ edges and
diameter $\ot(D+n^{1-\delta})$ with high probability. This construction uses $\ot(\min\{m, n^{1+\delta}\})$ messages and runs in $\ot(n^{1-\delta})$ rounds with high probability.
\end{theorem}

We need the following corollary of this theorem.
\begin{corollary}
\label{cor:danner}
Given an $n$-vertex, $m$-edge, diameter $D$ graph $G = (V, E)$ and a parameter $\delta \in [0, 1]$, there exists a randomized algorithm to solve the leader election and broadcast problems in the synchronous KT-1 \congest\ model using $\ot(\min\{m, n^{1+\delta}\})$ messages and in $\ot(D + n^{1-\delta})$ rounds with high probability.
\end{corollary}

We use this corollary to share $O(\poly \log n)$ random bits in a message-efficient manner by first
electing a leader and then having the leader locally generate the random bits and
broadcasting them. The message and time complexities
for this operation are given by the above corollary.
We note that the above danner bounds hold in \KTOne\ \congest\ model, which is \emph{synchronous}.
In the \emph{asynchronous} version of the \KTOne\ \congest\ model, we appeal to the following result.

\begin{theorem}[Mashregi and King \cite{KMDISC19,KMDISC18}]
\label{th:asyncst}
Given an $n$-vertex, $m$-edge graph $G = (V, E)$, there exists a randomized algorithm to construct a minimum spanning tree and (hence) solve the leader election and broadcast problems in the asynchronous KT-1 \congest\ model using $\ot(\min\{m, n^{1.5}\})$ messages and in $O(n)$ rounds, with high probability.
\end{theorem}

\section{Message Complexity Lower Bounds}

\subsection{Technical Preliminaries}
We now state key definitions and notation from Awerbuch et al.~ \cite{AwerbuchGPV1988} which we will use in our proofs of the $\Omega(m)$ message lower bounds for $(\Delta+1)$-coloring and MIS, for comparison-based algorithms, in the \KTOne\ \congest\ model.

\begin{definition}
[Executions]
We denote the execution of a \congest\ algorithm (or protocol) \(\mathcal{A}\) on a graph \(G(V, E)\) with an ID-assignment \(\phi\) by \(EX(\mathcal{A}, G, \phi)\). This execution contains (i) the messages sent and received by the nodes in each round and (ii) a snapshot of the local state of each node in each round.
We denote the state of a node \(v\) in the beginning of round \(i\) of the execution \(EX(\mathcal{A}, G, \phi)\) by \(L_i(EX, v)\).
\end{definition}

The \emph{decoded representation} of an execution is obtained by replacing each occurrence of an ID value \(\phi(v)\) by \(v\) in the execution. This decoded representation allows us to define a similarity of executions. We denote the decoded representations of all messages sent during round \(i\) of an execution \(EX(\mathcal{A}, G, \phi)\) as \(h_i(EX(\mathcal{A}, G, \phi))\).

\begin{definition}
[Similar executions]
Two executions of a \congest\ algorithm \(\mathcal{A}\) on graphs \(G(V, E)\) and \(G'(V, E')\) with ID-assignments \(\phi\) and \(\phi'\) are \emph{similar} if they have the same decoded representation. Likewise, we say that two messages are \emph{similar} if their decoded representations are the same.
\end{definition}

A crucial element of our lower bound proof consists of taking two graphs $G(V, E)$ and $G'(V', E')$, where $G'$ is obtained from $G$ by ``crossing'' a pair of edges in $G$, and showing that the executions of any comparison-based algorithm, on $G$ and $G'$ are similar.
Showing similarity of executions requires that the ``crossing'' of edges remains, in a certain
sense, hidden from the algorithm. Below, we define what it means for an algorithm to \emph{utilize} an edge. Later on we will be able to show that if the edges being ``crossed'' are not utilized by the algorithm, then the edge ``crossing'' is hidden from the algorithm.
One way an algorithm utilizes an edge is by sending a message across it. But, this notion of utilization does not suffice in the KT-\(1\) model. We need the stronger notion, defined below.
\begin{definition}
[Utilized Edge]
An edge \(e=\{u, v\}\) is utilized if any one of the following happens during the course of the algorithm:
(i)  a message is sent along \(e\),
(ii) the node $u$ sends or receives \texttt{ID} $\phi(v)$, or
(iii) the node $v$ sends or receives \texttt{ID} $\phi(u)$.
\end{definition}

By definition, the number of utilized edges is an upper bound on the number of edges along which a message sent. Using a charging argument, Awerbuch et al.~\cite{AwerbuchGPV1988} show that the number of utilized edges is also upper bounded by a constant times the number of edges along which a message sent. We restate their claim here.

\begin{lemma}[Lemma 3.4 of \cite{AwerbuchGPV1988}]
\label{lem:utilization-message-complexity}
Let \(m_u\) denote the number of utilized edges in an execution \(EX(\mathcal{A}, G, \phi)\). Then the message complexity of the execution is \(\Omega(m_u)\).
\end{lemma}

\subsection{Lower Bound Graph Construction and \texttt{ID} Assignments}
\label{sec:kt-1-lb}
We now describe the construction of lower bound graphs that we use for our \(\Omega(n^2)\) message complexity lower bounds.
The same construction works for both the $(\Delta+1)$-coloring and MIS lower bounds.
Recall that these bounds are for comparison-based algorithms in the KT-1 \congest\ model.

We start with a graph \(G(X, Y, Z, E)\)
such that \(|X| = |Y| = |Z| = t\) and the subgraphs of $G$ induced by $X \cup Y$ and $Y \cup Z$ are both isomorphic to the complete bipartite graph $K_{t, t}$. Thus, \(|E| = 2t^2\).
We then add a copy \(G'(X', Y', Z', E')\) of $G$ and consider the graph $G \cup G'$. We call
this the \emph{base graph}.
Let \(V = X \cup Y \cup Z\) and \(V' = X' \cup Y' \cup Z'\). For each \(v \in V\), the corresponding copy in \(V'\) is named \(v'\).
Let $n = |V \cup V'|$. Thus $t = n/6$.
From the base graph $G \cup G'$, we obtain a  \emph{crossed graph} as follows.
For a vertex \(y \in Y\), \emph{cross} an edge \(e = \{y,z\}\) in \(G\), where \(z \in Z\) with the edge \(e' = \{x',y'\}\) in \(G'\) where \(x' \in X'\) to obtain the graph \(G_{e,e'}\). When we cross the edge \(e = \{y, z\} \in E\) with \(e' = \{x', y'\} \in E'\), the resulting \emph{crossed graph} \(G_{e,e'}\) has vertex set \(V \cup V'\) and edge set \((E \cup E' \setminus \{e, e'\}) \cup \{\{y, y'\}, \{x', z\}\}\). The base graph $G \cup G'$ and the crossed graph \(G_{e,e'}\) for edges \(e \in E, e' \in E'\) are illustrated in Figure \ref{fig:lower-bound-graph}.

\begin{figure}
\begin{center}
\includegraphics[width=0.8\linewidth]{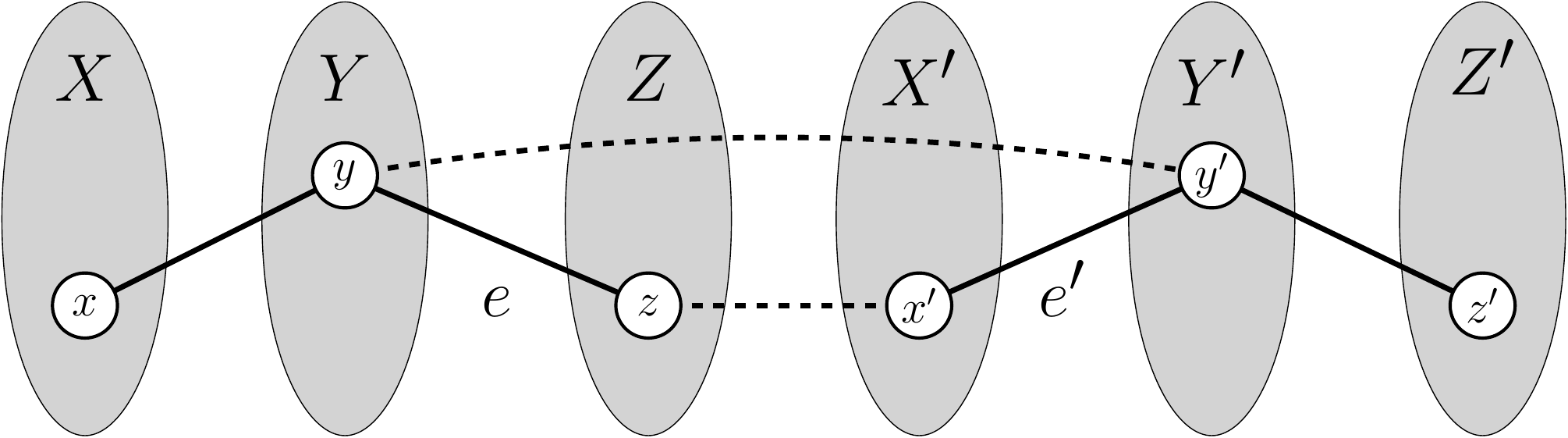}
\end{center}
\caption{\label{fig:lower-bound-graph} This figure shows the base graph $G \cup G'$ and the crossed graph \(G_{e,e'}\), described in Section \ref{sec:kt-1-lb}}
\end{figure}

We now define appropriate ID-assignments for the base graph and the crossed graph.
Let \(S\) be an arbitrary totally ordered set such that \(|S| = 40t\), and let \(\overline{S}\) be the sorted list of elements in \(S\) in ascending order. We will assign distinct elements in $S$ as \texttt{ID}'s to the base graph and the crossed graph.
We use a short-hand and say that the \texttt{ID} of a vertex \(v\) is \(i \in [0, 40t)\), when we mean that the ID of \(v\) is \(\overline{S}[i]\). Note that since \(\overline{S}\) is sorted in ascending order, the relative ordering of the indices is the same as that of the corresponding \texttt{ID}'s in \(\overline{S}\).

Let \(\phi: V \to [0, 40t)\) be an \texttt{ID} assignment such that \(\phi(v)\) is even for all \(v \in V\) and additionally \(\phi(v) \in [0, 2t)\) if \(v \in X\), \(\phi(v) \in [10t, 12t)\) if \(v \in Y\), and \(\phi(v) \in [20t, 22t)\) if \(v \in Z\).
For a vertex $y \in Y$ and pair of incident edges $e =\{y, z\}$ and $e' = \{x',y'\}$, we
define a ``shifted'' \texttt{ID} assignment $\phi'_{e,e'}$ for the vertex set $V'$ of $G'$.
We motivate this ``shifted'' assignment and define it precisely further below.
But for now, assuming $\phi'_{e, e'}$ is defined, we
define the \texttt{ID} assignment \(\psi_{e,e'}: V \cup V' \to [0, 40t)\) as just the union of \(\phi\) and \(\phi'_{e,e'}\), i.e., \(\psi_{e,e'}(v) = \phi(v)\) for all \(v \in V\) and \(\psi_{e,e'}(v') = \phi'_{e,e'}(v')\) for all \(v' \in V'\).
Our first goal in this subsection is to show that these two executions
$$EX = EX(\mathcal{A}, G \cup G', \psi_{e,e'}) \qquad\qquad EX_{e,e'} = EX(\mathcal{A}, G_{e,e'}, \psi_{e,e'})
$$
on the base graph $G \cup G'$ and the crossed graph
$G_{e,e'}$ are similar for any comparison-based algorithm $\mathcal{A}$.

For the executions $EX$ and $EX_{e,e'}$ to be similar, it must be the case that the crossing of edges $e$ and $e'$ is hidden from algorithm $\mathcal{A}$.
To achieve this, the \texttt{ID} assignment $\phi'_{e,e'}$ of $V'$ must be carefully chosen.
For example, vertex $z$ has neighbor $y$ in $G \cup G'$, but has neighbor $x'$ in $G_{e,e'}$ (see Figure \ref{fig:lower-bound-graph}). In the KT-1 model, $z$'s initial local knowledge consists of vertex $y$ in $G \cup G'$ and vertex $x'$ in $G_{e,e'}$. Therefore, for $\mathcal{A}$ to not distinguish between these two situations, it must be the case that the \texttt{ID} of $x'$ is ``adjacent'' to the \texttt{ID} of $y$. To achieve this, without disrupting other constraints on the relative order of \texttt{ID}'s, we start by assigning vertices in $X'$ the \texttt{ID}'s of their corresponding vertices in $X$ and then ``shift'' these by $(\phi(y) - \phi(x)) + 1$. As a result, vertex $x'$ ends up with \texttt{ID} $\phi(y)+1$. A similar ``shift'' is performed to obtain the \texttt{ID}'s of vertex set $Y'$, though this time the ``shift'' is by the amount $(\phi(z) - \phi(y))+1$ because we want vertex $y'$ to be ``adjacent'' to vertex $z$. The
``shift'' for vertex set $Z'$ just needs to be so that the \texttt{ID} assignment is disjoint,
We now define the \texttt{ID} assignment \(\phi'_{e,e'}: V' \to [0, 40t)\) as
\begin{equation}
\label{eqn:shiftedAssignment}
\phi'_{e,e'}(v') = \begin{cases}
                    \phi(v) + (\phi(y) - \phi(x)) + 1, \mbox{ if }v' \in X'  \\
                    \phi(v) + (\phi(z) - \phi(y))+ 1, \mbox{ if }v' \in Y' \\
                    \phi(v) + 10t + 1, \mbox{ if }v' \in Z'
                 \end{cases}
\end{equation}
Note that the \texttt{ID}s of
all vertices in each of the parts, $X'$, $Y'$, and $Z'$, are ``shifted'' by the same amount, though \texttt{ID}s in different parts may be ``shifted'' by different amounts.

The following observations about \(\phi'_{e,e'}\) are easy to verify.
\begin{itemize}
\item[(i)] The ranges of \(\phi\) and \(\phi'_{e,e'}\) are disjoint.
\item[(ii)] Moreover, \(\phi'_{e,e'}(v) \in [8t + 1, 14t+1]\) if \(v \in X'\), \(\phi'_{e,e'}(v) \in [18t+1, 24t+1]\) if \(v \in Y'\), and \(\phi'_{e,e'}(v) \in [30t+1, 32t+1]\) if \(v \in Z'\).
\item[(iii)] For any $u, v \in V$, $u \not= v$,
$\phi(u) < \phi(v)$ iff $\phi'_{e,e'}(u') < \phi'_{e,e'}(v')$.
\end{itemize}
Item (iii) is simply saying that the \texttt{ID} ordering on $V'$ induced by \(\phi'_{e,e'}\) is the same as the \texttt{ID} ordering induced by \(\phi\) with respect to the corresponding vertices in $V$.
This follows from the fact that the \texttt{ID}'s of  vertices in $X'$ are obtained by shifting the \texttt{ID}'s of vertices in $X$ by the same amount, thus preserving the relative ordering of \texttt{ID}'s in $X$ and $X'$. Similarly, for vertex sets $Y'$ and $Z'$. Furthermore, even though the \texttt{ID}'s of different sets, $X'$, $Y'$, and $Z'$ are obtained by ``shifting'' by different amounts, the ``shifting'' also ensures that \texttt{ID}'s in $X'$ remain less than \texttt{ID}'s in $Y'$, which in turn remain less than \texttt{ID}'s in $Z'$.

To prove that $EX$ and $EX_{e, e'}$ are similar, we need two intermediate \texttt{ID} assignments
for the set $V \cup V'$.
Recall that edge $e = \{y, z\}$ and edge $e' = \{x', y'\}$.
\begin{itemize}
\item[(i)] Define \(\psi_{e,e',x}\) to be the \texttt{ID} assignment \(\psi_{e,e'}\) except for interchanging the values of \(x'\) and \(y\) (i.e. \(\psi_{e,e',x}(y) = \phi'_{e,e'}(x')\) and \(\psi_{e,e',x}(x') = \phi(y)\)).
\item[(ii)] Define \(\psi_{e,e',z}\) analogously as \(\psi_{e,e'}\) except for interchanging the values of \(y'\) and \(z\) (i.e. \(\psi_{e,e',z}(z) = \phi'_{e,e'}(y')\) and \(\psi_{e,e',z}(y') = \phi(z)\)).
\end{itemize}
Using these \texttt{ID} assignments, we define two
intermediate executions on the base graph $G \cup G'$.
$$EX_{e,e',x} = EX(\mathcal{A}, G \cup G', \psi_{e,e',x}); \quad
  EX_{e,e',z} = EX(\mathcal{A}, G \cup G', \psi_{e,e',z})
$$
The following lemma, which shows that $EX$, $EX_{e,e',x}$, and $EX_{e,e',y}$ are similar, critically uses the fact that the \texttt{ID} assignment $\psi_{e,e'}$ shifts the \texttt{ID}'s of vertices in $X' \cup Y' \cup Z'$ so that the \texttt{ID} of $x'$ becomes ``adjacent'' to the \texttt{ID} of $y$ and the
\texttt{ID} of $y'$ becomes ``adjacent'' to the \texttt{ID} of $z$.
\begin{lemma}\label{lem:swapping-similarity}
For any $x \in X$, $y \in Y$, $z \in Z$ and edges $e = \{y, z\}$ and $e'=\{x', y'\}$, the executions \(EX\), \(EX_{e,e',x}\) and \(EX_{e,e',z}\) are similar.
\end{lemma}
\begin{proof}
All three executions have the same input graph $G \cup G'$. The execution pair \(EX\) and \(EX_{e,e',x}\) have the same \texttt{ID} assignment except for the vertices \(x'\) and \(y\), which have their \texttt{ID}'s swapped. Note that by definition of $\psi'_{e,e'}$ and \(\phi'_{e,e'}\) in (\ref{eqn:shiftedAssignment}), we have
$$\psi_{e,e'}(x') = \phi'_{e,e'}(x') = \phi(x) + (\phi(y) - \phi(x)) + 1 = \phi(y) + 1.$$
Furthermore, $\psi_{e,e'}(y) = \phi(y)$.
Therefore, when we swap the \texttt{ID}'s of $x'$ and $y$ in $\psi_{e,e'}$ to obtain $\psi_{e,e',x}$, there is no change in the relative ordering of \texttt{ID}'s and therefore
the executions \(EX\) and \(EX_{e,e',x}\) are similar.

A similar argument holds for the execution pair \(EX\) and \(EX_{e,e',z}\). By the definition of $\psi'_{e,e'}$ and \(\phi'_{e,e'}\) in (\ref{eqn:shiftedAssignment}), we have
$$\psi_{e,e'}(y') = \phi'_{e,e'}(y') = \phi(y) + (\phi(z) - \phi(y)) + 1 = \phi(z) + 1$$
and \(\psi_{e,e'}(z) = \phi(z)\).
Thus the relative ordering of \texttt{ID}'s in $\psi_{e,e'}$ and $\psi_{e,e',z}$ is the same and therefore the executions \(EX\) and \(EX_{e,e',x}\) are similar.

The lemma follows because similarity of executions is transitive.
\end{proof}

We can derive the final tool we need by directly appealing to a lemma in Awerbuch et al.~\cite{AwerbuchGPV1988}. Informally, the lemma shows that if edges $e = \{y,z\}$ and $\{x',y'\}$ are not utilized in the execution $EX$ of algorithm $\mathcal{A}$, then executions $EX$ and $EX_{e,e'}$ are similar. The main obstacle is that the initial knowledge vertices \(x', y, y', z\) is different in $EX$ and $EX_{e,e'}$ so a direct inductive proof like in Lemma \ref{lem:swapping-similarity} does not work. But we can use the intermediate executions of Lemma \ref{lem:swapping-similarity} to show the similarity for these four vertices. For all other vertices, we can do a direct inductive argument.

\begin{lemma}[Restatement of Lemma 3.8 of \cite{AwerbuchGPV1988}]
\label{lem:utilization-similarity}
Let $x \in X$, $y \in Y$, and $z \in Z$ be arbitrary vertices and let $e = \{y, z\}$ and $e' = \{x', y'\}$.
Suppose that during the first \(r\) rounds of the execution \(EX\) both $e$ and $e'$ are not utilized. Then the following hold for every round \(1 \le i \le r\) of the executions \(EX\), \(EX_{e,e',x}\) , \(EX_{e,e',z}\) and \(EX_{e,e'}\):
\begin{enumerate}
\item The states of the nodes in the beginning of the round, i.e. \(L_i(\cdot, \cdot)\) satisfy:
\begin{enumerate}
\item For every processor \(w \in V \setminus \{y, z, y', x'\}\), \(L_i(EX_{e,e'}, w) = L_i(EX, w)\).
\item For \(u \in \{x', z\}\), \(L_i(EX_{e,e'}, u) = L_i(EX_{e,e',x}, u)\).
\item For \(v \in \{y, y'\}\), \(L_i(EX_{e,e'}, v) = L_i(EX_{e,e',z}, v)\).
\end{enumerate}
\item The messages sent during the round are similar, i.e., \(h_i(EX) = h_i(EX_{e,e',x}) = h_i(EX_{e,e',z}) = h_i(EX_{e,e'})\).
\item In \(EX_{e,e'}\), no messages are sent during the round over the edges \(\{x', z\}\) and \(\{y, y'\}\).
\end{enumerate}
\end{lemma}

\begin{corollary}
\label{cor:utilization-similarity}
Suppose that during the execution \(EX\) neither of the edges \(e = \{y,z\}\) and \(e' = \{x',y'\}\) are utilized, for some vertices \(x \in X\), \(y \in Y\), and \(z \in Z\). Then the executions \(EX\) and \(EX_{e,e'}\) are similar and furthermore in \(EX_{e,e'}\), no messages are sent through the edges \(\{y,y'\}\) and \(\{x',z\}\).
\end{corollary}
In the next subsections, we will show that this similarity leads to a contradiction with
respect to correctness for problems such as $(\Delta+1)$-coloring and MIS. This in turn will
imply a constraint on the behavior of algorithm $\mathcal{A}$: for every pair of edges \(e = \{y,z\}\) and \(e' = \{x',y'\}\), at least one of the edges is utilized by $\mathcal{A}$. This in turn will
lead to the message complexity lower bound we desire.

\subsection{$\Omega(m)$ message lower bound for \((\Delta + 1)\)-Coloring in \KTOne\ \congest}
Now that we have shown that $EX$ and $EX_{e,e'}$ are similar if $e$ and $e'$ are not utilized by algorithm $\mathcal{A}$, we will show that for some problems this leads to a contradiction.
The intuition for this is simple.
Let $\phi$ and $\phi'$ be \texttt{ID} assignments for $V$ and $V'$ respectively, that consistently order the vertices, i.e., $\phi(u) < \phi(v)$ iff $\phi'(u') < \phi'(v')$ for all $u, v \in V$. Since $G$ and $G'$ are isomorphic, it is easy to show that $EX_G = EX(\mathcal{A}, G, \phi)$ and $EX_{G'} = EX(\mathcal{A}, G', \phi')$ are similar. This is shown below in Lemma \ref{lem:id-assignment-similarity} below. Now consider the base graph $G \cup G'$ and the \texttt{ID} assignment $\psi_{e,e'}$ of $V \cup V'$. Lemma \ref{lem:id-assignment-similarity} implies that corresponding vertices $v$ and $v'$ have the same local states after execution $EX = EX(\mathcal{A}, G \cup G', \psi_{e,e'})$ completes. Since $EX$ and $EX_{e,e'} = EX(\mathcal{A}, G_{e,e'}, \psi_{e,e'})$ are similar, this also implies that vertices $v$ and $v'$ have the same local states after execution $EX_{e,e'}$. But, in the crossed graph $G_{e, e'}$, $y$ and $y'$ are neighbors. For problems in which neighboring vertices ought not to have the same local state (e.g., neighboring vertices cannot have the same color in a solution to the vertex coloring problem), this is a contradiction.

\begin{lemma}\label{lem:id-assignment-similarity}
Consider an arbitrary vertex $y \in Y$ and an arbitrary pair of edges $e =\{y, z\}$, $z \in Z$ and $e' = \{x',y'\}$, $x' \in X'$.
For any comparison-based algorithm $\mathcal{A}$ in the KT-1 \congest\ model,
the executions \(EX_G = EX(\mathcal{A}, G, \phi)\) and \(EX_{G'} = EX(\mathcal{A}, G', \phi'_{e,e'})\) are similar.
\end{lemma}
\begin{proof}
Since the input graphs \(G\) and \(G'\) are copies of each other, the only thing that is different between the two executions is the \texttt{ID} assignments.
However, Property (iii) of the \texttt{ID} assignment \(\phi'_{e,e'}\) above implies that every \texttt{ID} comparison by $\mathcal{A}$ on $G$ yields the same result as the corresponding
\texttt{ID} comparison on $G'$.
Therefore, by an inductive argument it can be shown that at the beginning of each round, the state of each vertex \(v\) in \(G\) is the same as the state of the corresponding vertex \(v'\) in \(G'\) and the messages received by these vertices are also be the same. This gives us that the executions \(EX_{G}\) and \(EX_{G'}\) are similar.
\end{proof}

\begin{lemma}
\label{lem:coloring-lb-correctness}
Let \(x \in X\), \(y \in Y\), and \(z \in Z\) be three vertices such that the edges \(e = \{y,z)\}\) and \(e' = \{x',y'\}\) are not utilized in the execution \(EX\). Then, algorithm \(\mathcal{A}\) computes an incorrect \((\Delta + 1)\)-coloring for the crossed graph $G_{e,e'}$.
\end{lemma}
\begin{proof}
In the execution \(EX\), since the input graph has two disconnected components \(G\) and \(G'\), Lemma~\ref{lem:id-assignment-similarity} gives us that the color of a vertex \(v\) in \(G\) is the same as the color of the corresponding vertex \(v'\) in \(G'\). Since the edges \(e = \{y,z\}\) and \(e' = \{x',y'\}\) are not utilized in \(G \cup G'\), applying Corollary \ref{cor:utilization-similarity}, \(\mathcal{A}\) will compute the same coloring in the graph \(G_{e,e'}\) as it will in \(G \cup G'\). This implies a monochromatic edge \(\{y,y'\}\) in \(G_{e,e'}\) which contradicts the correctness of the algorithm.
\end{proof}

\begin{theorem}
[Deterministic Lower Bound]
Let \(\mathcal{A}\) be a deterministic comparison-based algorithm that computes a \((\Delta+1)\)-coloring. Then the message complexity of \(\mathcal{A}\) is \(\Omega(n^2)\). This holds even if the vertices know the size of the network.
\end{theorem}
\begin{proof}
Suppose that \(\mathcal{A}\) is a deterministic comparison-based algorithm that computes a \((\Delta+1)\)-coloring and has message complexity \(o(n^2)\). Then by Lemma \ref{lem:utilization-message-complexity}, the number of edges utilized by $\mathcal{A}$ is $o(n^2)$. This implies that there exists a $y \in Y$ and edges $e = \{y, z\}$ and $e' = \{x', y'\}$ such that $e$ and $e'$ are not utilized when $\mathcal{A}$ executes on $G \cup G'$. By Lemma \ref{lem:coloring-lb-correctness} this implies that $\mathcal{A}$ computes an incorrect $(\Delta+1)$-coloring for $G_{e,e'}$.
\end{proof}
We now extend this lower bound to Monte Carlo randomized algorithms, even with constant error probability. To do this we strengthen Lemma \ref{lem:coloring-lb-correctness} so that it applies not just to a single crossed graph, but to the entire family of crossed graphs. Let \(\mathcal{F}\) denote the family of all crossed graphs, i.e.,
$\mathcal{F} = \{G_{e, e'} \mid e = \{y, z\}, e'=\{x', y'\}, x, y, z, \in V\}$.
Note that $|\mathcal{F}| = t^3$ because there are $t$ choices for $y$ and for each choice of $y$, there are $t^2$ choices for $e$ and $e'$.
\begin{lemma}
\label{lem:extend-deterministic-lb-coloring}
Let \(\mathcal{A}\) be a deterministic comparison-based \KTOne\ \congest\ algorithm that computes a \((\Delta+1)\)-coloring correctly on at least a constant \(\delta\) fraction of graphs in the family \(\mathcal{F}\). Then the message complexity of \(\mathcal{A}\) is \(\Omega(\delta n^2)\). This holds even if the vertices know the size of the network.
\end{lemma}
\begin{proof}
Assume for the sake of contradiction that the message complexity of \(\mathcal{A}\) is \(o(\delta n^2)\). By Lemma \ref{lem:utilization-message-complexity}, we have that \(\mathcal{A}\) utilizes \(o(\delta n^2)\) edges in any graph that it runs on. Specifically consider the execution \(EX\) of algorithm \(\mathcal{A}\) on input graph \(G \cup G'\) and ID assignment \(\psi_{e,e'}\) where \(e,e'\) denote a graph \(G_{e,e'}\) in the family \(\mathcal{F}\).

Since \(\mathcal{A}\) utilizes \(o(\delta n^2)\) edges, there can only be \(o(n) = o(t)\) vertices in \(Y\) such that more than \(cn/6 = ct\) incident edges are utilized, for some constant \(c\) to be determined later. Recall that \(t = n/6\). The rest of the \(t - o(t)\) vertices in \(Y\) have less than \(ct\) incident edges that are utilized. By Lemma~\ref{lem:id-assignment-similarity} the same statement holds for the corresponding vertices in \(Y'\) because in \(EX\), the two graphs \(G\) and \(G'\) that form the input graph are disconnected, which implies the executions of \(\mathcal{A}\) on \(G\) and \(G'\) are similar.

So for each such vertex \(y \in Y\), there are at most \((c^2/4)t^2\) edge pairs of the form \(e = \{y, z\},e' = \{x', y'\}\) such that \(e,e'\) are utilized. Therefore, by Lemma \ref{lem:coloring-lb-correctness}, the algorithm computes an incorrect \((\Delta+1)\)-coloring on at least \((1 - o(1))(1 - (c^2/4)) = 1 - (c^2/4) - o(1)\)-fraction of the graphs in \(\mathcal{F}\) (since for each \(y \in Y\) there are exactly \(t^2\) graphs in \(\mathcal{F}\)). Setting \(c = \sqrt{2\delta}\), the algorithm computes an incorrect \((\Delta+1)\)-coloring on at least \(1 - \delta/2 - o(1)\)-fraction of the graphs in \(\mathcal{F}\). This is a contradiction if  \(1 - \delta < 1 - \delta/2 - o(1)\) or \(\delta > o(1)\). Since \(\delta\) is a constant, we get a contradiction.
\end{proof}

A simple application of Yao's lemma \cite{YaoFOCS1977,MotwaniRaghavan} with the uniform distribution on all the graphs in the family \(\mathcal{F}\) gives the following theorem.

\begin{theorem}
[Randomized Lower Bound]
Let \(\mathcal{A}\) be a randomized Monte-Carlo comparison based \KTOne\ \congest\ algorithm that computes a \((\Delta + 1)\)-coloring with probability of error less than a constant \(\epsilon \in [0, 1)\). Then the worst case message complexity of \(\mathcal{A}\) is \(\Omega((1-\epsilon)n^2)\). This holds even if the vertices know the size of the network.
\end{theorem}

\subsection{$\Omega(m)$ message lower bound for MIS in \KTOne\ \congest}

\begin{lemma}
\label{lem:MIS-correctness}
Let \(x \in X\), \(y \in Y\), and \(z \in Z\) be three vertices such that the edges \(e = \{y,z\}\) and \(e' = \{x',y'\}\) are not utilized in the execution \(EX\). Then, algorithm \(\mathcal{A}\) computes an incorrect MIS on $G_{e,e'}$.
\end{lemma}
\begin{proof}
Due to Lemma~\ref{lem:id-assignment-similarity}, we can only have two MIS's in \(G \cup G'\) which are: \(Y, Y'\) or \(X, X', Z, Z'\). Since the edges \(e = \{y,z\}\) and \(e' = \{x',y'\}\) are not utilized in \(G \cup G'\), applying Corollary \ref{cor:utilization-similarity}, \(\mathcal{A}\) will compute the same MIS in the graph \(G_{e,e'}\) as it will in \(G \cup G'\). Both the MIS solutions mentioned above violate the independence requirement of MIS in \(G_{e,e'}\). This contradicts the correctness of the algorithm.
\end{proof}

A similar proof as the coloring deterministic lower bound gives us the following theorem.
\begin{theorem}
[Deterministic Lower Bound]
Let \(\mathcal{A}\) be a deterministic comparison-based \KTOne\ \congest\ algorithm that solves the MIS problem. Then the message complexity of \(\mathcal{A}\) is \(\Omega(n^2)\). This holds even if the vertices know the size of the network.
\end{theorem}
The following lemma for MIS parallels Lemma \ref{lem:extend-deterministic-lb-coloring} for $(\Delta+1)$-coloring. This has a very similar proof, which we skip.
\begin{lemma}
Let \(\mathcal{A}\) be a deterministic comparison-based \KTOne\ \congest\ algorithm that computes an MIS correctly on at least a constant \(\delta\) fraction of graphs in the family \(\mathcal{F}\). Then the message complexity of \(\mathcal{A}\) is \(\Omega(\delta n^2)\). This holds even if the vertices know the size of the network.
\end{lemma}

A simple application of Yao's lemma \cite{YaoFOCS1977,MotwaniRaghavan} with the uniform distribution on all the graphs in the family. \(\mathcal{F}\) gives the following theorem.

\begin{theorem}
[Randomized Lower Bound]
Let \(\mathcal{A}\) be a randomized Monte-Carlo comparison-based \KTOne\ \congest\ algorithm that computes an MIS with probability of error less than a constant \(\epsilon \in [0, 1)\). Then the worst case message complexity of \(\mathcal{A}\) is \(\Omega((1-\epsilon)n^2)\). This holds even if the vertices know the size of the network.
\end{theorem}

\subsection{$\Omega(n)$ message lower bound in \KTRho\ \congest}
The $\Omega(m)$ lower bounds we have proved apply to comparison-based algorithms in the \KTOne \congest\ model. We now prove a weaker $\Omega(n)$ message complexity bound for $(\Delta+1)$-coloring and MIS, but these apply more generally, to all algorithms (even non-comparison-based) and to the \KTRho \congest\ model, for any constant $\rho$.

\begin{theorem} \label{thm:ktrho}
  Any randomized Monte Carlo algorithm that computes an MIS or a $(\Delta+1)$-vertex coloring with probability at least $\tfrac{5}{8}$, requires $\Omega(n)$ messages in expectation in the \KTRho $\congest$ model, for any constant $\rho$.
\end{theorem}

\begin{proof}
  Similarly to \cite{DBLP:journals/siamdm/Naor91,DBLP:journals/siamcomp/Linial92}, we assume without loss of generality that algorithms follow the general framework that all nodes perform their coin flips initially and only exchange their current local state (including coin flips) without performing any other local computation until the very last round.

  For the given constant $\rho$, define the constant $k$ to be the smallest integer such that
  \[
    \log^*(k) \ge 2(\rho + 3).
  \]
  Consider an $n$-node graph $G$ consisting of the disjoint union of $n/k$ cycles each of $k$ nodes.\footnote{For simplicity, we assume that $n/k$ is an integer.}
  For each cycle $C_i$, we fix a set of IDs $R_i$ from some integer range of size $k$ such that all ID ranges assigned to the cycles are pairwise disjoint.
  We will equip the nodes of each cycle $C_i$ with $k$ unique IDs, as described below.

  Consider any randomized algorithm $B_0$ that works on a cycle.
  We know from \cite{DBLP:journals/siamdm/Naor91} that any randomized algorithm that succeeds in computing a $3$-coloring on cycle $C_i$ with probability more than $\tfrac{1}{2}$ requires at least $\tfrac{1}{2}\log^*(k) - 2$ rounds in the worst case, under the \KTZero assumption.
  Now suppose that we execute $B_0$ on our lower bound graph $G$ for at most $\rho < \tfrac{1}{2}\log^*(k) - 2$ rounds under the \KTZero assumption.
  Even though $B_0$ is not guaranteed to work on $G$, it nevertheless exhibits some well-defined behavior on each cycle that we will exploit.
  Observe that each node $u$ in $G$ is part of some cycle $C_i$ and hence the color output by $u$ is a function of the observed neighborhood and the random coin flips.
  We point out that, even though that $u$ also has knowledge of $n$, it is easy to see that this does not have any impact on the output of the algorithm.
  As a straightforward consequence of \cite{DBLP:journals/siamdm/Naor91}, we know that, for each cycle $C_i$, there exists some \emph{hard ID assignment} $I_i$, which is a permutation of the set $R_i$, such that $B_0$ fails to yield a valid coloring on $C_i$ with some probability greater than $\tfrac{1}{2}$ (assuming \KTZero), where this probability is taken over the coin flips of the nodes in $C_i$.

  Returning to the \KTRho assumption, suppose towards a contradiction that there exists an algorithm $B_\rho$ that computes a $3$-coloring on $G$ while sending $o(n)$ messages in expectation.
  We provide additional power to the algorithm by revealing, to each node $u$, the coin flips of the nodes in its $\rho$-neighborhood.

  We assign the IDs of the nodes in each cycle $C_i$ according to $I_i$.
  Since there are $n/k = \Omega(n)$ cycles but the expected message complexity of $B_\rho$ is $o(n)$, it holds that, with probability at least $\tfrac{3}{4}$, there exists a cycle $C_j$ such that the nodes in $C_j$ do not send any messages at all when executing $B_\rho$; call this event $\textsc{Mute}$.
  We now condition on $\textsc{Mute}$ occurring:
  Consider any node $u \in C_j$ and observe that the output of $u$ is a function of its initial knowledge, i.e., its random coin flips and the local state of its $\rho$-neighborhood.
  Clearly, the behavior of $u$ follows the exact same probability distribution when executing $B_\rho$ under the \KTRho assumption as it does when executing algorithm $B_0$ under the \KTZero assumption.
  In particular, the result of \cite{DBLP:journals/siamdm/Naor91} implies that some neighboring nodes in $C_j$ will output the same color with probability greater than $\tfrac{1}{2}$.
  Since event $\textsc{Mute}$ occurs with probability at least $\tfrac{3}{4}$, it follows that algorithm $B_\rho$ fails with probability $>\tfrac{3}{8}$, yielding a contradiction.
\end{proof}

\section{Upper bounds in KT-1 \congest}

\subsection{\((\Delta + 1)\)-Coloring using \(\tilde{O}(n^{1.5})\) Messages in KT-1 \congest}

In this section we present a $(\Delta+1)$-list-coloring algorithm in the KT-1 \textsc{Congest} model that uses \(\tilde{O}(n^{1.5})\) messages. This algorithm is obtained by utilizing -- with some modifications -- the
simple graph partitioning technique introduced recently by Chang et al.~\cite{ChangFGUZPODC2019}.
This technique is central to the fast $(\Delta+1)$-coloring algorithms that Chang et al.~\cite{ChangFGUZPODC2019} obtain in different models of computation, namely Congested Clique, MPC, and Centralized Local Computation.

The Chang et al.~\cite{ChangFGUZPODC2019} graph partitioning scheme is as follows. Let $\Psi(v)$ denote the palette of vertex $v \in V$ and let $k = \sqrt{\Delta}$.
\begin{itemize}
\item\textbf{Vertex set partition:} We partition $V$ into $B_1,\dots, B_{k}, L$ as follows.  Include each $v \in V$ in the set $L$ with probability $q = \Theta\left( \frac{\sqrt{\log n }}{\Delta^{1/4}}\right)$.
Then each remaining vertex joins one of $B_1 , \ldots, B_{k}$ uniformly at random.

\item\textbf{Palette partition:} Let $C = \bigcup_{v \in V} \Psi(v)$ be the set of all colors. We partition $C$  into $k$ sets $C_1, \dots, C_{k}$ where each color $c \in C$ joins one of the $k$ sets uniformly at random.
\end{itemize}

Chang et al.~\cite{ChangFGUZPODC2019} then show that whp, the output of the partitioning scheme satisfies the following 4 properties, assuming that $\Delta = \omega(\log^2 n)$. These properties allow us to color each set $B_i$ using palette $C_i$, in parallel, and then recursively color the set $L$ until it becomes small enough to color trivially.

\begin{description}
\item[(i) Size of Each Part:] $|E(G[B_i])| = O(|V|)$, for each $i \in [k]$. Also,  $|L| = O(q |V|) = O\left(\frac{\sqrt{\log n}}{\Delta^{1/4}}\right) \cdot |V|$.

\item[(ii) Available Colors in $B_i$:] For each $i\in \{1, \ldots,k\}$ and $v \in B_i$, let the number of available colors in $v$ in the subgraph $B_i$ be $g_i(v) := |\Psi(v) \cap C_i|$. Then $g_i(v) \geq \Delta_i+1$, where
$\Delta_i := \max_{v \in B_i}\deg_{B_i}(v)$.

\item[(iii) Available Colors in $L$:] For each $v \in L$, define $g_L(v) := |\Psi(v)| - (\deg_G(v) - \deg_L(v))$. It is required that $g_L(v) \geq \max\{\deg_L(v), \Delta_L - \Delta_L^{3/4}\}+1$ for each $v \in L$, where $\Delta_L := \max_{v \in L}\deg_{L}(v)$. Note that $g_L(v)$ represents a lower bound on the number of available colors in $v$'s palette {\em after } all of $B_1, \ldots, B_k$ have been colored.

\item[(iv) Remaining Degrees:] The maximum degrees of $B_i$ and $L$ are $\deg_{B_i}(v)\leq \Delta_i = O(\sqrt{\Delta})$ and $\deg_{L}(v) \leq \Delta_L = O(q\Delta) =  O\left(\frac{\sqrt{\log n}}{\Delta^{1/4}}\right) \cdot \Delta$. For each vertex, we have that $\deg_{B_i}(v)\leq \max\{O(\log n), O(1/\sqrt{\Delta}) \cdot \deg(v)\}$ and also have $\deg_{L}(v)\leq \max\{O(\log n), O(q) \cdot \deg(v)\}$.
\end{description}

 We now present our algorithm, which takes as input an $n$-vertex graph $G$ with maximum degree $\Delta$ and diameter $D$. The algorithm runs in the KT-1 \congest\ model and produces a $(\Delta+1)$-list-coloring of $G$ using $\ot(n^{1.5})$ messages and running in
$\ot(D + \sqrt{n})$ rounds.
\RestyleAlgo{boxruled}
\begin{algorithm2e}\caption{KT-\(1\) \((\Delta+1)\)-Coloring Algorithm:\label{alg:KT1-Delta+1-coloring}}
For $\delta = 1/2$, build a danner $H$, elect a leader $\ell$, and have the leader broadcast a string $R$ of $O(\log^2 n)$ random bits. \\
Nodes use the $O(\log^2 n)$ bits of $R$ to sample three $O(\log n)$-wise independent hash functions: (a) $h_L$, to decide whether to join $L$, (b) $h$, to decide which set $B_i$ to join, and (c) $h_c$, to determine which color goes into which part $C_i$.\\
Nodes execute a randomized algorithm for list coloring by Johansson~\cite{Johansson99} in each \(B_i\) in parallel.\\
Using the danner $H$, we can check whether the induced graph $G[L]$ has \(\tilde{O}(n)\) edges.\\
If it does, we execute the list coloring algorithm by Johansson~\cite{Johansson99} on $G[L]$.\\
If not, we recursively run this algorithm on \(G[L]\) with the same parameter \(n\).
\end{algorithm2e}

The ``full independence'' version of the following lemma is proved in~\cite{ChangFGUZPODC2019}. We provide a brief sketch of the changes required in this proof to make a version with limited independence go through.
\begin{lemma}
\label{lemma:graphPartitioningProperties}
Properties (i)-(iv) mentioned above hold w.h.p.,
even when the partitioning of vertices and colors is done using $O(\log n)$-wise independence, as described in Line 2 of Algorithm \ref{alg:KT1-Delta+1-coloring}.
\end{lemma}
\begin{proof}
Chang et al.~\cite{ChangFGUZPODC2019} show that this lemma holds when the vertex partitioning is done using full independence, while the color partitioning is done using $O(\log n)$-wise independence. A closer look at their proof reveals that all four properties are shown using Chernoff bounds, and these bounds can be safely replaced by limited dependence Chernoff bounds described in Lemma~\ref{lemma:Chernoff2}. Therefore the four properties hold whp even when the partitioning of both vertices and colors is done using $O(\log n)$-wise independence.
\end{proof}

The following lemma is proved in~\cite{ChangFGUZPODC2019} and given that Properties (i)-(iv) hold in the limited independence setting we use, it goes through without any changes.
\begin{lemma}
\label{lem:O(1)-recursive-calls}
The algorithm makes $O(1)$ recursive calls w.h.p.
\end{lemma}

\begin{theorem}
Given as input an $n$-vertex graph $G$ with maximum degree $\Delta$ and diameter $D$, Algorithm \ref{alg:KT1-Delta+1-coloring} runs in the KT-1 \congest\ model and produces a $(\Delta+1)$-list-coloring of $G$ using $\ot(n^{1.5})$ messages and running in
$\ot(D + \sqrt{n})$ rounds.
\end{theorem}
\begin{proof}
In Step 1, we create a danner with the parameter \(\delta = 1/2\), building the danner takes \(\ot(\sqrt{n})\) rounds and $\ot(n^{1.5})$ messages, see Lemma~\ref{lem:danner}. Also, because of danner property, the diameter of the danner
is $\ot(D+\sqrt{n})$ and hence electing a leader and broadcasting a \(\log^2 n\)-bit random string
takes $\ot(D+\sqrt{n})$ rounds and $\ot(n^{1.5})$ messages (Corollary~\ref{cor:danner}).  Step 2 is just local computation.

In step 3 we use Johansson's randomized algorithm on each \(G[B_i]\). This algorithm works in \(O(\log n)\) rounds and \(\ot(|E(G[B_i])|)\) messages whp even when the palette of each vertex \(v\) has been initialized to an arbitrary subset of \(\deg(v)+1\) colors chosen from \(\{1, 2, \dots, \Delta+1\}\)~\cite{Johansson99}. According to Property (ii), the palettes of vertices in each \(B_i\) are large enough. And according to property (i), \(|E(G[B_i])| = |V|\) So this step runs in \(O(\log n)\) rounds and takes \(\ot(n\sqrt{\Delta})\) messages over all the \(B_i\)'s whp, since there are \(O(\sqrt{\Delta})\) \(B_i\)'s.

Step 4 takes \(\ot(D + \sqrt{n})\) rounds and $\ot(n^{1.5})$ messages by Corollary~\ref{cor:danner}.  Using Lemma~\ref{lem:O(1)-recursive-calls} and the above arguments guarantee that steps 5 and 6 together take \(\ot(D + \sqrt{n})\) rounds and \(\ot(n^{1.5})\) messages whp. The theorem follows.
\end{proof}

\subsubsection{Asynchronous KT-1 \congest\ algorithm}

The $(\Delta+1)$-coloring in the \congest\ KT-1 mode described above (Algorithm \ref{alg:KT1-Delta+1-coloring}) has a natural counterpart in the asynchronous version of the \congest\ KT-1 model.
The broadcast of random bits in Step 1 can be done asynchronously using $\ot(n^{1.5})$ messages and in $O(n)$ rounds by appealing to the result of Mashregi and King \cite{KMDISC19,KMDISC18} (see Theorem \ref{th:asyncst}).
Every node, on receiving the random bits that were broadcast, completes Step 2 via local computation.
Due to shared randomness, each node $v \in B_i$, for each $i \in [k]$, knows which of its neighbors are in $B_i$. The coloring algorithm in Step 3 therefore can be executed by nodes in $B_i$ by communicating just over the edges in the induced graph $G[B_i]$.
The synchronous algorithm used in Step 3 in Algorithm \ref{alg:KT1-Delta+1-coloring} can be simulated by using an $\alpha$-synchronizer \cite{AwerbuchJACM1985} (see Theorem \ref{theorem:alphaSynchronizer}) in an asynchronous setting. This takes $\ot(n)$ messages since (i) according to Lemma \ref{lemma:graphPartitioningProperties}, $G[B_i]$ contains $O(n)$ edges and (ii) Step 3 runs in $O(\log n)$ rounds.
Checking if $G[L]$ has $\ot(n)$ edges (Step 4) can be done via asynchronous upcast, using $\ot(n)$ messages and
in $O(n)$ rounds. This is possible because each node $v \in L$, knows which of its neighbors are in $L$ and can therefore send the size of its $L$-restricted neighborhood up the spanning tree.
Like Step 3, Step 5 can also be executed using $O(n)$ messages, in $O(\log n)$ rounds.
This description leads to the following theorem.

\begin{theorem}
Given as input an $n$-vertex graph $G$ with maximum degree $\Delta$, there is an algorithm that runs in the asynchronous KT-1 \congest\ model and produces a $(\Delta+1)$-list-coloring of $G$ using $\ot(n^{1.5})$ messages and running in
$\ot(n)$ rounds.
\end{theorem}

\subsection{\((1 + \epsilon)\Delta\)-Coloring using \(\tilde{O}(n)\) Messages in KT-1 \congest}

In this section, we show that for any $\epsilon > 0$, there is an algorithm that can compute a $(1 + \epsilon)\Delta$-coloring in the KT-1 \textsc{Congest} model in \(\tilde{O}(n)\) rounds, using \(\tilde{O}(n/\epsilon^2)\) messages.

At the beginning of the algorithm,
for a large enough constant $C$, one node generates $(C/\epsilon) \cdot \log^3 n$ random bits and shares it with all other nodes using a danner \cite{GmyrPanduranganDISC18}, using $\tilde{O}(n/\epsilon)$ messages and $\tilde{O}(n)$ rounds in the
KT-1 \textsc{Congest} model (cf. Corollary \ref{cor:danner}).
In the following algorithm, each node $v$ that has not already permanently colored itself, will use random bit string $s_i$ in Phase $i$ to first select a random hash function $h_i$ from a
 family of $\Theta(\log n)$-wise independent hash functions $\mathcal{H} = \{h: [\text{poly}(n)] \to [(1+\epsilon)\Delta]\}$.
 Node $v$ will then compute $h_i(\texttt{ID}_v)$ to pick a random color from the palette $[(1+\epsilon)\Delta]$.
 Note that the length of $s_i$ is $\Theta(\log^2 n)$ and by Lemma \ref{lemma:randomBits}, this number of random bits suffice to pick a $\Theta(\log n)$-wise independent hash function with domain size $\text{poly}(n)$ and range size $(1+\epsilon)\Delta$.
 In Corollary \ref{cor:KT1Termination}, it is shown that Algorithm \ref{alg:KT1-coloring} runs in $O(\log n/\epsilon)$ phases and therefore $r = \Theta(\log n/\epsilon)$ random bit strings suffice.

\RestyleAlgo{boxruled}
\begin{algorithm2e}\caption{\((1 + \epsilon)\Delta\)-Coloring Algorithm (One phase):\label{alg:KT1-coloring}}
Each active node (i.e., which has not been colored yet) chooses a random candidate color from \((1+\eps) \Delta\) color palette. \\

It makes this color permanent if it is sure that none of its neighbors has chosen this color yet. \\

If unsuccessful in choosing a permanent color, go to step 1.
\end{algorithm2e}

In step 2, we will show that a node has to check only a small subset of its neighbors in any phase.

First, we will show that the probability of success in each phase is large.

\begin{lemma}
\label{lem:color-progress}
In any phase, a node chooses a color that has not been chosen by any of its neighbors in this phase or in any previous phases with probability at least \(\eps/(1+\eps) \approx \eps\) (for small \(\eps\)). Hence there will be no conflict with the chosen color and hence the node will successfully color itself. Thus, a node successfully colors itself in \(O(\log n/\eps)\) rounds whp.
\end{lemma}

\begin{proof}
Fix a node \(v\) and suppose that it has degree $d$. Arbitrarily label the neighbors of $v$, $v_1, v_2, \ldots, v_d$ and let $X_i$ be the indicator variable that indicates if neighbor $v_i$ has picked same color as $v$.
Let $X = \sum_{i=1}^d X_i$.
Then $\Pr[X_i = 1]=1/(1+\epsilon)\Delta$ and $\mathbb{E}[X] \le d/(1+\epsilon)\Delta \le 1/(1+\epsilon)$. Then, by Markov's inequality, $\Pr[X \ge 1] \le 1/(1+\epsilon)$.
Therefore, $\Pr[X = 0] \ge \epsilon/(1+\epsilon) \approx \eps$, (for small \(\eps > 0\)).
Since $X = 0$ represents the event that no neighbor of $v$ chooses the color $v$ picked, we get the first part of the lemma.

Thus after \((c\log n)/\eps\) rounds for a large enough constant \(c\), \(v\) will successfully color itself whp.
\end{proof}

\begin{corollary}
\label{cor:KT1Termination}
Whp, all nodes successfully color themselves in \(O(\log n/\eps)\) rounds.
\end{corollary}

\begin{proof}
By  Lemma~\ref{lem:color-progress} and union bound over all nodes.
\end{proof}

Implementing step 2 with small message complexity:

\begin{lemma}
In each phase, each node exchanges at most \(O(\log^2 n/\eps)\) messages whp.
\end{lemma}
\begin{proof}
In Step 2, a node will check to see if the color chosen by itself is not chosen by any of its neighbors. To check this,
it will only check neighbors that could have chosen this color in this round or in any previous rounds.
Fix a node $v$ and let $c$ be the color it samples in this round.
Arbitrarily label the neighbors of $v$, $v_1, v_2, \ldots, v_d$ and let $X_i$ be the indicator variable that indicates if neighbor $v_i$ has picked same color as $v$ in this round.
Let $X = \sum_{i=1}^d X_i$.
Then $\Pr[X_i=1]=1/(1+\epsilon)\Delta$ and $\mathbb{E}[X] \le d/(1+\epsilon)\Delta \le 1/(1+\epsilon)$.
Since the colors of vertices are chosen using an $\Theta(\log n)$-wise independent family of hash functions, the variables $X_1, X_2, \ldots, X_d$ are $\Theta(\log n)$-wise independent.
Then, by Lemma \ref{lemma:Chernoff2}, for a sufficiently large constant $A$, $\Pr[X \ge A\cdot \log n] \le \exp(-2\log n) = 1/n^2$.
Therefore, whp there are at most $O(\log n)$ neighbors of $v$ that could have picked
color $c$ in this round.

This is true of color $c$ in previous rounds as well. Node $v$ has to check all these neighbors which have chosen $c$ in this round
or prior rounds to be sure that there is no conflict in choosing $c$.
Since there are at most \(O(\log n/\eps)\) phases whp (by Lemma~\ref{lem:color-progress}), color $c$ is chosen by only \(O(\log n \log n/\eps) = O(\log^2 n/\eps)\)
neighbors whp.
\end{proof}

\begin{theorem}
There is a coloring algorithm that achieves \((1+\eps)\Delta\) coloring using \(O(n\log^3n/\eps^2)\) messages whp in KT1 model (with shared randomness).
\end{theorem}

\begin{proof}
By Lemma~\ref{lem:color-progress}, all nodes can legally color themselves in \(O(\log n/\eps)\) rounds whp.
By Lemma 2, each node exchanges \(O(\log^2 n/\eps)\) messages in a phase and there are at most \(O(\log n/\eps)\) phases.
Hence the overall message complexity (of all nodes) is  \(O(n\log^3 n/\eps^2)\) whp.
\end{proof}

\section{An MIS algorithm using $\ot(n^{1.5})$ messages in KT-\(2\) \congest}
\label{sec:mis-alg-kt-2}

We now give a high-level overview of Algorithm~\ref{alg:KT2-mis} that uses \KTTwo knowledge to compute an MIS using only \(O(n^{1.5} \log^2 n)\) messages  while taking \(\ot(\sqrt{n})\) rounds; the full details are explained in the proof of Theorem~\ref{thm:KT2-mis}.
We first sample a set $S$ of $\Theta(\sqrt{n})$ nodes and then add these nodes to the independent set according to the randomized greedy MIS algorithm. Since $S$ was chosen randomly, this has the same effect as performing $\Theta(\sqrt{n})$ iterations of the sequential randomized greedy algorithm, which is known to reduce the maximum degree in the remnant graph to $\tilde O(\sqrt{n})$ (see \cite{KonradArxiv2018}).
Then, each node $u \in S$ that entered the independent set informs its 2-hop neighbors.
It is crucial that node $u$ uses its \KTTwo knowledge to convey this information, as otherwise the same 2-hop neighbor $v$ might receive $u$'s message from multiple 1-hop neighbors of $u$, which may result in $\omega(n)$ messages being sent on behalf of $u$.
Finally, we compute an MIS on the (sparsified) remnant graph using Luby's algorithm.

\RestyleAlgo{boxruled}
\begin{algorithm2e}\caption{MIS Algorithm \label{alg:KT2-mis}}

\textbf{Sample \(O(\sqrt{n})\) vertices:} Add each node to a set \(S\) with probability \(c/\sqrt{n}\).\\

\textbf{Run Randomized Greedy MIS:} Each node in \(S\) chooses a random rank at the start of the algorithm. In the parallel version of Greedy, a node enters the MIS as soon as it is a local maximum among undecided neighbors in \(S\).

\textbf{Inform \(2\)-hop Neighbors:} Each node \(u \in S\) that enters the MIS \(u\) uses KT-2 knowledge to inform all of its 2-hop neighbors that it has joined the MIS.

\textbf{Pruning Inactive Edges:} Each node \(v \in V\) uses its own KT-\(2\) knowledge to either deactivate itself if a \(1\)-hop neighbor has joined the MIS or deactivate edges incident on the \(1\)-hop neighbors that are neighbors with a node that joins the MIS. \\

\textbf{Finishing Up:} All nodes in the remnant graph know which of their neighbors are deactivated and so we can run Luby's algorithm on the remnant graph.
\end{algorithm2e}

\begin{theorem} \label{thm:KT2-mis}
Algorithm~\ref{alg:KT2-mis} computes a correct MIS. It uses \(O(n^{1.5} \log^2 n)\) messages and runs in
 \(\ot(\sqrt{n})\) rounds  with high probability.
\end{theorem}
\begin{proof}
In the first two steps, we aim to simulate \(O(\sqrt{n})\) iterations of the sequential randomized greedy algorithm, see Algorithm 2 in \cite{BlellochFSSPAA2012}. In the randomized greedy MIS algorithm, each node chooses a random rank at the start of the algorithm and we process the nodes by rank to compute an MIS using the greedy algorithm. Simulating \(i\) iterations of this algorithm is probabilistically equivalent to sampling \(i\) vertices uniformly at random and generating random ranks at just tho sampled vertices. Note that since we sample vertices uniformly at random with probability \(c / \sqrt{n}\), we get \(|S| = O(\sqrt{n})\) whp.

We instead run the parallel randomized greedy MIS algorithm which computes the same MIS as the sequential version (see \cite{BlellochFSSPAA2012}), and in \cite{FischerNSODA2018}, they show that the parallel randomized greedy MIS algorithm finishes in \(O(\log n)\) rounds whp, and so the message complexity of steps  will just be \(O(|S|n\log n) = O(n^{1.5} \log n)\) whp.

Using KT-2 information, each vertex in \(u \in S\) that joins the MIS locally creates a depth \(2\) BFS tree on its \(2\)-hop neighborhood and sends the message through this tree. This BFS tree is constructed by having all \(1\)-hop neighbors of \(u\) at depth \(1\) and assigning a node \(v\) that is exactly \(2\)-hops away as a child to the \(1\)-hop neighbor with lowest ID. The local view of this BFS tree can be created at all \(1\)-hop neighbors of \(u\) using their own KT-2 information since the common \(1\)-hop neighbors of \(u\) and \(v\) are all \(2\)-hops away.

To send the message to \(1\)-hop neighbors, \(u\) can just broadcast. The one hop neighbors will just inform their neighbors in the BFS tree that \(u\) has joined the MIS. In case a node \(w\) gets multiple messages of \(1\)-hop neighbors in \(S\) joining the MIS and their BFS trees lead to the same \(2\)-hop vertex \(v\), then \(w\) will just send all these messages one by one to \(v\). The congestion on such an edge can be at most  \(|S|\) in the worst case. This allows each such node \(v\) to prune the inactive edges and learn the edges of the remnant graph that are incident on it without sending or receiving any additional messages.

Since \(|S| = O(\sqrt{n})\) whp, and each vertex in \(S\) that joins the MIS can relay this information to it's 2-hop neighbors using constant messages per neighbor, this process generates at most \(O(|S|n) = O(n^{1.5} \log n)\) messages whp. But due to congestion, this process will take \(O(\sqrt{n})\) whp in the worst case.

After the simulation, we know from Lemma 1 in \cite{KonradArxiv2018} that the remnant graph has maximum degree \(O(n\log n / |S|) = O(\sqrt{n} \log n)\). And since the nodes know the remnant graph, running Luby's algorithm~\cite{LubySTOC1985} will require an additional \(O(\log n)\) rounds and \(O(n^{1.5} \log^2 n)\) messages whp.

Therefore, Algorithm~\ref{alg:KT2-mis} runs in \(O(\sqrt{n})\) rounds and uses \(O(n^{1.5} \log^2 n)\) messages throughout its execution whp. The theorem follows.
\end{proof}

\section{Conclusion}
In this paper, we initiate the study of the message complexity of two fundamental symmetry breaking problems, MIS and $(\Delta+1)$-coloring.
We show that while it is impossible to obtain $o(m)$ message complexity in the \KTOne \congest\ model using comparison-based algorithms, one can do so by either using non-comparison based algorithms or by slightly increasing the input knowledge, i.e., in the \KTTwo \congest\ model.

Several key open questions arise from our work.
The first is whether one can obtain an $o(m)$-message, non-comparison-based algorithm for MIS in the \KTOne \congest\ model, running in polynomial time.
We have shown that this is possible for $(\Delta+1)$-coloring.
The second is whether one can obtain (nearly optimal) $\ot(n)$-message (non-comparison-based) algorithms
for MIS and $(\Delta+1)$-coloring in the \KTOne\ \congest\ model, running in polynomial time.
The question is open for MIS even in the \KTTwo\ \congest\ model. Another important issue is reducing the running time of our algorithms. In particular, can we make them run in $\polylog{n}$ rounds, for the same message bounds?

\bibliographystyle{ACM-Reference-Format}
\bibliography{references}

\newpage
\appendix
\section{Appendix}
\subsection{Tail inequalities and hash functions with limited independence}

To obtain message-efficient algorithms in the \KTOne{} model, we make use
of hash functions with limited independence. These hash functions use $c$-wise
independence and hence we use the following tail inequalities and properties
of such hash functions.

\noindent
The following tail inequalities are from \cite{SchmidtSSSODA1993}.
\begin{lemma}
\label{lemma:Chernoff1}
Let $c \ge 4$ be an even integer. Suppose
$Z_1, Z_2, \ldots, Z_t$ are $c$-wise
independent random variables taking values in
$[0, 1]$. Let $Z = \sum_{i=1}^t Z_i$ and
$\mu = \mathbb{E}[Z]$, and let $\lambda > 0$.
Then,
$$Pr[|Z - \mu|\ge \lambda] \le 2\left(\frac{ct}{\lambda^2}\right)^{c/2}.$$
\end{lemma}
\begin{lemma}
\label{lemma:Chernoff2}
Suppose that $X$ is the summation of
$n$, $c$-wise independent 0-1 random variables, each with mean $p$.  Let $\mu$ satisfy $\mu \ge \mathbb{E}[X] = np$. Then,
$$\Pr[X \ge (1 + \delta)\mu] \le exp(-\min\{c, \delta^2\mu\}).$$
\end{lemma}

\noindent
The following is Definition 7 in \cite{CzumajDPArxiv2020}.
\begin{definition}
For $N$, $L$, $c \in \mathbb{N}$, such that $c \le N$, a family of functions $\mathcal{H} = \{h: [N] \to [L]\}$ is $c$-wise independent if for all distinct $x_1,x_2,\ldots,x_c
\in [N]$, the random variables $h(x_1), h(x_2), \ldots, h(x_c)$ are independent and uniformly distributed in
$[L]$ when $h$ is chosen uniformly at random from
$H$.
\end{definition}

\noindent
The following lemma appears as Corollary 3.34 in \cite{VadhanFTTCS2012}.
\begin{lemma}
\label{lemma:randomBits}
For every $a,b,c$, there is a family of $c$-wise independent hash functions $\mathcal{H} = \{h:\{0,1\}^a \to \{0,1\}^b\}$
such that choosing a random function from $\mathcal{H}$
takes $c\cdot \max\{a, b\}$ random bits, and evaluating a function from $\mathcal{H}$ takes $\text{poly}(a, b, c)$
computation.
\end{lemma}

\subsection{Simulating synchronous algorithms in an asynchronous model}

\begin{theorem}[Awerbuch's $\alpha$-synchronizer \cite{AwerbuchJACM1985}]
\label{theorem:alphaSynchronizer}
Given a synchronous Algorithm $A$ running in $T$ rounds on a graph with $m$ edges in the \KTRho \congest\ model for any $\rho \ge 1$, it is possible to simulate $A$ in the asynchronous \KTRho \congest\ model in $T$ rounds.
The  number of additional messages sent in the asynchronous execution compared to an execution of
$A$ is at most $2(T+1)m$.
\end{theorem}

 \end{document}